\documentclass{SciPost}

\hypersetup{
    colorlinks,
    linkcolor={red!50!black},
    citecolor={blue!50!black},
    urlcolor={blue!80!black}
}

\usepackage[bitstream-charter]{mathdesign}
\usepackage{amsthm}
\usepackage{bm}
\DeclareSymbolFont{usualmathcal}{OMS}{cmsy}{m}{n}
\DeclareMathAlphabet\mathbfcal{OMS}{cmsy}{b}{n}
\DeclareSymbolFontAlphabet{\mathcal}{usualmathcal}

\fancypagestyle{SPstyle}{
\fancyhf{}
\lhead{\colorbox{scipostblue}{\bf \color{white} ~SciPost Physics Core}}
\rhead{{\bf \color{scipostdeepblue} ~Submission }}

\fancyfoot[C]{\textbf{\thepage}}
}

 \usepackage{hyperref}

\renewcommand{\Re}{\operatorname{Re}}
\renewcommand{\Im}{\operatorname{Im}}

\newtheorem{theorem}{Theorem}
\newtheorem{lemma}{Lemma}
\newtheorem{corollary}{Corollary}
\newtheorem{definition}{Definition}
\newtheorem{proposition}{Proposition}

\begin{document}

\pagestyle{SPstyle}

\begin{center}{\Large \textbf{\color{scipostdeepblue}{
On truncations of hierarchical equations of motion for finite-dimensional systems\\
}}}\end{center}

\begin{center}\textbf{
Vasilii Vadimov\textsuperscript{1 $\star$}
}\end{center}

\begin{center}
{\bf 1}
QCD Labs, QTF Centre of Excellence, Department of Applied Physics, Aalto University, P.O. Box 15100, FI-00076 Aalto, Finland\\
$\star$ \href{mailto:vasilii.1.vadimov@aalto.fi}{\small vasilii.1.vadimov@aalto.fi}
\end{center}

\section*{\color{scipostdeepblue}{Abstract}}
\textbf{
    We study truncations of hierarchical equations of motion (HEOM) for finite-dimensional open quantum systems. We prove that for finite-dimensional approximations constructed with a Schur-complement-type terminator, the spectrum converges to that of the full HEOM as the truncation depth increases. We also prove that this approximation is free of spectral pollution: sufficiently deep truncations do not produce spurious unstable modes, provided the exact HEOM is stable. We illustrate the results for the spin-boson model.
}

\vspace{\baselineskip}

\noindent\textcolor{white!90!black}{%
\fbox{\parbox{0.975\linewidth}{%
\textcolor{white!40!black}{\begin{tabular}{lr}%
  \begin{minipage}{0.6\textwidth}%
    {\small Copyright attribution to authors. \newline
    This work is a submission to SciPost Physics Core. \newline
    License information to appear upon publication. \newline
    Publication information to appear upon publication.}
  \end{minipage} & \begin{minipage}{0.4\textwidth}
    {\small Received Date \newline Accepted Date \newline Published Date}%
  \end{minipage}
\end{tabular}}
}}
}


\vspace{10pt}
\noindent\rule{\textwidth}{1pt}
\tableofcontents
\noindent\rule{\textwidth}{1pt}
\vspace{10pt}

\section{Introduction}

Hierarchical equations of motion (HEOM)~\cite{Tanimura1989, Tanimura2020} are a well-established and widely used method for non-perturbative numerical analysis of open quantum system dynamics. The method has a broad range of applications~\cite{Tanimura2020}, including quantum chemistry~\cite{Bai2024}, quantum transport~\cite{Yang2025}, correlated electronic systems~\cite{Ye2016, Dan2023}, and other areas. Due to its conceptual simplicity and versatility, HEOM has become a standard tool in computational quantum physics and several open-source implementations are currently available online~\cite{Struempfer2012, Lambert2023, Huang2023, Zhang2024}.

In the HEOM formalism, non-Markovian features of the environment are resolved by introducing an infinite hierarchy of auxiliary density operators (ADOs)~\cite{Tanimura2020, Xu2023, Xu2026a} in addition to the physical reduced density operator of the system. This raises a number of questions about the mathematical structure of HEOM: whether the dynamics preserves complete positivity~\cite{Witt2017}, whether it is sensitive to the particular rational approximation of the bath correlation function~\cite{Xu2022, Huang2024}, and which metric is natural in the ADO hierarchy space~\cite{Chen2024a}. Recent work~\cite{Mueller2026} has shown that HEOM is equivalent to a pseudomode Lindblad master equation, provided that the rational approximation of the bath correlation function preserves its non-negativity. This outstanding result addresses two of the above issues: the resulting Lindbladian dynamics is completely positive and contractive with respect to the trace norm of the density operator.

The infinite hierarchy creates a practical problem: numerical calculations require a finite truncation. Naive truncations can suffer from spectral pollution and become unstable~\cite{Shi2018, Dunn2019, Yan2020a, Li2022, Krug2023}. Several approaches have been proposed to mitigate this issue, including projection techniques~\cite{Dunn2019} and continuous-variable representations of the hierarchy~\cite{Yan2020a, Li2022, Chen2024a}.

In Ref.~\cite{Krug2023}, the authors focused specifically on truncation stability and examined different truncation schemes and coupling strengths. They concluded that ``truncating the hierarchy to any finite size can be problematic for strong coupling to a dissipative reservoir, in particular when combined with an appreciable reservoir memory time''. In this work, we would like to follow up that study and explore stability of truncated HEOM rigorously. We prove that, provided that the truncated finite-dimensional Liouvillian is properly constructed, its spectrum converges to that of the exact HEOM as the truncation size increases. As a consequence, if the exact HEOM is itself free of intrinsic instabilities, sufficiently deep truncations are also stable. We do \emph{not}, however, address the question of HEOM stability as such.

The paper is organized as follows. In Sec.~\ref{sec:HEOM}, we introduce the HEOM formalism and fix the notations used throughout the paper. In Sec.~\ref{sec:spectral}, we study spectral properties of the full HEOM Liouvillian. We use the fact that the inter-hierarchy couplings grow slower than the diagonal dissipation terms and derive bounds on the Liouvillian resolvent inspired by Gershgorin's theorem~\cite{Gershgorin1931, Varga2011}. We then use a Schur-complement reduction~\cite{Griesemer2008} with respect to the tail of the hierarchy to analyze the resolvent inside the corresponding Gershgorin sets. In Sec.~\ref{sec:truncation}, we present the main results of the paper: we introduce a finite-dimensional approximation to the HEOM Liouvillian and prove two main theorems, on spectral convergence of the approximation and on its stability. Finally, in Sec.~\ref{sec:discussion}, we discuss the scope and limitations of the results and present an illustrative example, followed by conclusions in Sec.~\ref{sec:conclusions}.

To improve readability, most proofs of lemmas and theorems are presented in the appendices. We keep, however, the proof of Theorem~\ref{th:spectrum-convergence} in the main text, since it constitutes one of the two central results of the paper.

\section{HEOM formulation}
\label{sec:HEOM}

HEOM is formulated for a hierarchy of ADOs \(\hat \rho_{\bm n}\), where \(\bm n \in \mathbb N_0^N\) is a multiindex of \(N\) non-negative integers. The top of the hierarchy, \(\hat \rho_{\bm 0}\), coincides with the reduced density operator of the system. The direct physical meaning of the higher ADOs is less transparent, but they encode system--reservoir correlations. HEOM is linear and can be written as
\begin{equation}
    \frac{\mathrm d \hat \rho_{\bm n}}{\mathrm d t}
    =
    \sum\limits_{\bm m} \mathcal L_{\bm n \bm m} \hat \rho_{\bm m},
    \label{eq:HEOM0}
\end{equation}
where each Liouvillian component \(\mathcal L_{\bm n \bm m}\) is a linear superoperator acting on trace-class density operators. We restrict ourselves to the finite-dimensional systems, therefore each of these maps is bounded.

We use a symmetrically scaled version of HEOM~\cite{Shi2009, Zhu2011}, the structure of its Liouvillian $\mathbfcal L$ reads as (\(\hbar = 1\))
\begin{equation}
    \begin{aligned}
        \mathcal L_{\bm n \bm n}\hat \rho
        &=
        -i \left[\hat H, \hat \rho\right]
        - \sum\limits_{j=1}^N n_j \gamma_j \hat \rho, \\
        \mathcal L_{\bm n,\, \bm n + \bm e_j}\hat \rho
        &=
        c_j \sqrt{n_j + 1}\, \left[\hat q_j, \hat \rho\right],~j = 1,\ldots, N, \\
        \mathcal L_{\bm n,\, \bm n - \bm e_j}\hat \rho
        &=
        \sqrt{n_j} \left(
        c_j' \hat q_j \hat \rho + c_j'' \hat \rho \hat q_j
        \right),~j = 1,\ldots, N,
    \end{aligned}
    \label{eq:HEOM}
\end{equation}
with all the other entries of the Liouvillian vanishing.
Here, \(\hat H\) is the Hamiltonian of the system, \(\gamma_j\) are complex-valued decay rates with positive real parts, \(\bm e_j\) is the multiindex whose entries are zero except for the \(j\)-th entry, which is equal to \(1\), \(\hat q_j\) are coupling operators, and \(c_j\), \(c_j'\), and \(c_j''\) are complex-valued coupling constants, with \(j = 1, \ldots, N\). The complex-valued decay rates \(\gamma_j\) are determined by the poles of the bath correlation functions, while the coupling constants \(c_j\), \(c_j'\), and \(c_j''\) originate from the corresponding residues. We do not restrict ourselves to any particular choice of system or reservoir parameters and work with the general form of HEOM.

One can associate the multiindex \(\bm n\) with the occupation numbers of \(N\) auxiliary bosonic modes~\cite{Xu2023, Li2024, Chen2024a, Vadimov2025, Xu2026a}. In this picture, the coupling between hierarchy levels is similar to the action of bosonic ladder operators.

\subsection{Notations and norm choice}

A natural choice of norm for density operators is the trace norm,
\begin{equation}
    \|\hat \rho \| := \operatorname{Tr} \sqrt{\hat \rho^\dagger \hat \rho}.
\end{equation}
For operators in the Hilbert space of the system, we also use the standard spectral norm,
\begin{equation}
    \left \| \hat A \right \|
    :=
    \sup\limits_{\langle \psi | \psi \rangle = 1}
    \sqrt{\langle \psi | \hat A^\dagger \hat A | \psi \rangle }.
\end{equation}
It is equal to the largest singular value of \(\hat A\). For Hermitian operators, this norm coincides with the largest absolute value of the eigenvalues.

We denote superoperators acting on density operators by calligraphic capital letters, such as \(\mathcal A\). For them, we use the norm induced by the trace norm:
\begin{equation}
    \| \mathcal A \|
    :=
    \sup\limits_{\| \hat \rho \| = 1} \| \mathcal A \hat \rho\|.
\end{equation}

We denote a hierarchy of ADOs by \(\hat{\bm \rho}\). Its components \(\hat \rho_{\bm n}\) are indexed by elements of \(\mathbb N_0^N\) or by elements of its subsets. 
Our work is based on a generalization of Gershgorin's theorem~\cite{Gershgorin1931} and for ordinary matrices it can be conveniently formulated and proved in terms of $1$-norm and $\infty$-norm.
Therefore, we use the \(1\)-norm
\begin{equation}
    \| \hat{\bm \rho} \|_1
    :=
    \sum\limits_{\bm n} \| \hat \rho_{\bm n} \|,
\end{equation}
since it is more restrictive: if \(\|\hat {\bm \rho} \|_1 < +\infty\), then \(\| \hat{\bm \rho} \|_\infty < +\infty\), while the converse does not hold.

For linear maps \(\mathbfcal A\) acting between hierarchies of ADOs, we use the corresponding induced norm
\begin{align}
    \| \mathbfcal A \|_1
    &:=
    \sup \limits_{\| \hat {\bm \rho} \|_1 = 1}
    \sum_{\bm n} \left \| \sum\limits_{\bm m} \mathcal A_{\bm n \bm m} \hat \rho_{\bm m} \right \|
    \leqslant
    \sup_{\bm m} \sum\limits_{\bm n} \| \mathcal A_{\bm n \bm m} \|.
\end{align}
As one can see, the HEOM Liouvillian is unbounded.

Throughout the paper we use notations $\sigma(\mathcal A)$ and $\varrho(\mathcal A)$ for the spectrum and resolvent sets of the operator $\mathcal A$, respectively.

\section{Spectral properties of HEOM Liouvillian}

\label{sec:spectral}

In order to study the spectral properties of the Liouvillian, we analyze its resolvent
\begin{equation}
    \mathbfcal R(z;\mathbfcal L) := (z - \mathbfcal L)^{-1}.
\end{equation}
From the structure of the Liouvillian blocks~\eqref{eq:HEOM}, we see that the diagonal dissipation grows linearly with the hierarchy indices $n_j$, while the off-diagonal couplings grow only as $\sqrt{n_j}$. As a result, for certain \(z\), the infinite matrix \(z - \mathbfcal L\) becomes block-diagonally dominant with respect to the hierarchy multiindex. This allows us to identify a domain $\Omega(\mathbfcal L) \subset \mathbb C$ where the resolvent exists and to derive rigorous bounds on it.

To analyze the resolvent in the whole complex plane \(\mathbb C\), we split the hierarchy into a finite low-lying part and an infinite deep tail. We then use the obtained results for the resolvent of the tail block and study the spectrum of the Liouvillian through the Schur complement of~\(z - \mathbfcal L\).

\subsection{Gershgorin-type resolvent bound}

 
We generalize certain results by Gershgorin~\cite{Gershgorin1931} and subsequent works~\cite{Feingold1962, Varah1975, Salas1999, Varga2011} on the spectrum and resolvent of diagonally dominant matrices. Since we work with the \(1\)-norm, it is natural to employ column diagonal dominance.

\begin{definition}
    For a Liouvillian $\mathbfcal L$ and a multiindex $\bm n$ let
    \begin{equation}
        \beta(z; \mathcal L_{\bm{n n}})
        :=
        \inf_{\|\hat \rho\| = 1}
        \left\|
        (z - \mathcal L_{\bm n\bm n})\hat \rho
        \right\|,
    \end{equation}
    We define the \emph{column Gershgorin set} $G_{\bm n}(\mathbfcal L)$ as
    \begin{equation}
        G_{\bm n}(\mathbfcal L)
        :=
        \left\{
        z \in \mathbb C :
            \beta(z; \mathcal L_{\bm{n n}}) \leqslant R_{\bm n}(\mathbfcal L)
        \right\},
    \end{equation}
    where $R_{\bm n}(\mathbfcal L)
    :=
    \sum_{\bm m \ne \bm n}
    \|\mathcal L_{\bm m\bm n}\|$ is the \emph{Gershgorin radius}.
    \label{def:gershgorin-set}
\end{definition}
These Gershgorin sets generalize the Gershgorin circles~\cite{Gershgorin1931} to the case of operator-valued matrices~\cite{Salas1999}. We are now ready to formulate a theorem giving a resolvent bound.
\begin{theorem}
    Let
    $
    G(\mathbfcal L)
    :=
    \overline{\bigcup_{\bm n} G_{\bm n}(\mathbfcal L)}
    $
    be a closed union of all the Gershgorin sets and 
    $
    \Omega(\mathbfcal L)
    :=
    \mathbb C \setminus G(\mathbfcal L)
    $
    be its complement.
    Assume that for a given \(z \in \Omega(\mathbfcal L)\),
\begin{equation}
    q(z;\mathbfcal L)
    :=
    \sup_{\bm n}
    \frac{R_{\bm n}(\mathbfcal L)}{\beta(z; \mathcal L_{\bm{n n}})}
    < 1.
    \label{eq:qz-condition}
\end{equation}
Then \(z \in \varrho(\mathbfcal L)\), and
\begin{equation}
    \left\|
    \mathbfcal R(z;\mathbfcal L)
    \right\|_1
    \leqslant
    \sup_{\bm n}
    \left[
        \beta(z; \mathcal L_{\bm{nn}}) - R_{\bm n}(\mathbfcal L)
    \right]^{-1}.
    \label{eq:resolvent-bound-hard-clean}
\end{equation}
\label{th:gershgorin-hard}
\end{theorem}
Condition~\eqref{eq:qz-condition} does not follow automatically from $z \in \Omega(\mathbfcal L)$. In general, there may be a situation when $\beta(z; \mathcal L_{\bm{nn}}) \leqslant R_{\bm n}(\mathbfcal L) - \delta$ for some positive $\delta$ but $q(z; \mathbfcal L) = 1$. The requirement~\eqref{eq:qz-condition} is needed to prove existence of the resolvent.

We now adapt this theorem to the HEOM Liouvillian. We use the specific form of the Liouvillian~\eqref{eq:HEOM} to provide bounds for $\beta(z; \mathcal L_{\bm{nn}})$ and Gershgorin radii $R_{\bm n}(\mathbfcal L)$.
\begin{lemma}
    Let $\hat A$ be a Hermitian matrix and $b$ be a complex number. Let $\mathcal T$ be a linear map defined as
    $$\mathcal T \hat \rho := -i \left[\hat A, \hat \rho\right] + b \hat \rho.$$ 
    Let $$S\left(b, \hat A\right) := b + i\left[-\Delta\left(\hat A\right), \Delta\left(\hat A\right)\right], $$ where 
    $$\Delta\left(\hat A\right) := \max_{\lambda \in \sigma(\hat A)} \lambda - \min_{\lambda \in \sigma(\hat A)} \lambda.$$
    Then $\sigma(\mathcal T) \subset S\left(b, \hat A\right)$, and for $z \notin S(b, \hat A)$ the following bound holds
    \begin{equation}
        \beta(z; \mathcal T) \geqslant \operatorname{dist}\left(z, S\left(b, \hat A\right)\right).
    \end{equation}
    \label{lemma:inv-diagonal}
\end{lemma}
\begin{lemma}
    For a Liouvillian $\mathbfcal L$ given by~\eqref{eq:HEOM}, Gershgorin radii are bounded by
    \begin{equation}
        R_{\bm n}(\mathbfcal L)
        \leqslant
        \sqrt{C(\Re \gamma_{\bm n} + \gamma)},
        \label{eq:bound-off-diag}
    \end{equation}
    where
    \begin{equation}
        \begin{gathered}
            C := \sum\limits_{j=1}^N 
            \frac{\left \| \hat q_j \right \|^2}{
            \Re \gamma_j
            }
            \left(
            2 |c_j| + |c_j'| + |c_j''|
            \right) ^ 2,~
            \gamma_{\bm n} := \sum\limits_{j=1}^N n_j \gamma_j,~
            \gamma := \sum\limits_{j=1}^N \Re \gamma_j.
        \end{gathered}
    \end{equation}
    \label{lemma:gershgorin-radii}
\end{lemma}
These bounds are sufficient to show that the requirement~\eqref{eq:qz-condition} is fulfilled automatically for any HEOM Liouvillian given by~\eqref{eq:HEOM}. As a corollary, the statement of Theorem~\ref{th:gershgorin-hard} can be slightly simplified.

\begin{corollary}
    If the Liouvillian $\mathbfcal L$ is given by \eqref{eq:HEOM}, then $\Omega(\mathbfcal L) \subset \varrho(\mathbfcal L)$ and the bound
    \eqref{eq:resolvent-bound-hard-clean} holds for each $z \in \Omega(\mathbfcal L)$.
    \label{cor:resolvent-set}
\end{corollary}
For clarity, we move the proofs of the statements presented in this Section to Appendices~\ref{sec:th-gershgorin-hard-proof}, \ref{sec:lemma-inv-diagonal-proof}, \ref{sec:lemma-gershgorin-radii-proof}, and~\ref{sec:corr-resolvent-set}. 
\newcommand{\dist}{\operatorname{dist}}
%
%
%
%
%
%
%
\subsection{Block decomposition}

\newcommand{\LI}{{\mathbb T}}
\newcommand{\LJ}{{\overline {\mathbb T}}}
\newcommand{\LII}{{\LI\LI}}
\newcommand{\LIJ}{{\LI\LJ}}
\newcommand{\LJI}{{\LJ\LI}}
\newcommand{\LJJ}{{\LJ\LJ}}

The Gershgorin-type bound~\eqref{eq:resolvent-bound-hard-clean} provides useful control of the resolvent only outside the set \(G(\mathbfcal L)\). For \(z \in G(\mathbfcal L)\), this estimate by itself gives essentially no information. To study the spectrum in this region, we combine the Gershgorin-type estimate with a Schur-complement reduction.

\begin{definition}
    Let \(\LI \subset \mathbb N_0^N\) be a finite subset of hierarchy multiindices such that \(\bm 0 \in \LI\). We call such a subset a \emph{truncation of the hierarchy} and denote its complement by
    $ \LJ := \mathbb N_0^N \setminus \LI$.
\end{definition}

For a given truncation \(\LI\), we decompose the hierarchy into the sectors supported on~\(\LI\) and~\(\LJ\), and split the Liouvillian into four corresponding blocks:
\begin{equation}
    \mathbfcal L
    =
    \begin{bmatrix}
        \mathbfcal L_\LII & \mathbfcal L_\LIJ \\
        \mathbfcal L_\LJI & \mathbfcal L_\LJJ
    \end{bmatrix}.
\end{equation}
The resolvent of the Liouvillian then has the same block structure,
\begin{equation}
    \mathbfcal R(z; \mathbfcal L)  = 
    \begin{bmatrix}
        z - \mathbfcal L_\LII & -\mathbfcal L_\LIJ \\
        -\mathbfcal L_\LJI & z - \mathbfcal L_\LJJ
    \end{bmatrix}^{-1} = 
    \begin{bmatrix}
        \mathbfcal R_\LII(z; \mathbfcal L) & \mathbfcal R_\LIJ(z; \mathbfcal L) \\
        \mathbfcal R_\LJI(z; \mathbfcal L) & \mathbfcal R_\LJJ(z; \mathbfcal L)
    \end{bmatrix}.
    \label{eq:block-resolvent}
\end{equation}
For \(z \in \Omega\left(\mathbfcal L_\LJJ\right)\), the blocks of the resolvent are given by
\begin{equation}
\begin{aligned}
    \mathbfcal R_\LII(z; \mathbfcal L) &= \left[\mathbfcal S_\LI(z)\right]^{-1}, \\
    \mathbfcal R_\LIJ(z; \mathbfcal L) &= \mathbfcal R_\LII(z; \mathbfcal L) \mathbfcal L_\LIJ \left(z - \mathbfcal L_\LJJ\right)^{-1}, \\
    \mathbfcal R_\LJI(z; \mathbfcal L) &= \left(z - \mathbfcal L_\LJJ\right)^{-1} \mathbfcal L_\LJI \mathbfcal R_\LII(z; \mathbfcal L), \\
    \mathbfcal R_\LJJ(z; \mathbfcal L) &= \left(z - \mathbfcal L_\LJJ\right)^{-1} + 
    \left(z - \mathbfcal L_\LJJ\right)^{-1}
    \mathbfcal L_\LJI \mathbfcal R_\LII(z; \mathbfcal L) \mathbfcal L_\LIJ
    \left(z - \mathbfcal L_\LJJ\right)^{-1},
\end{aligned}
    \label{eq:resolvent-blocks}
\end{equation}
where
\begin{equation}
    \mathbfcal S_\LI(z) := z - \mathbfcal L_\LII - \mathbfcal L_\LIJ
    \left(z - \mathbfcal L_\LJJ\right)^{-1} \mathbfcal L_\LJI
\end{equation}
is the Schur complement of \(z-\mathbfcal L\) with respect to the block \(z-\mathbfcal L_\LJJ\). Since \(\LI\) is finite, \(\mathbfcal S_\LI(z)\) is a matrix-valued holomorphic function in \(\Omega\left(\mathbfcal L_\LJJ\right)\). A complex number \(\lambda \in \Omega\left(\mathbfcal L_\LJJ\right)\) belongs to the spectrum \(\sigma(\mathbfcal L)\) if and only if
$\det \mathbfcal S_\LI(\lambda) = 0$~\cite{Griesemer2008}.

\begin{definition}
    For a given truncation \(\LI\), we define \emph{lower decay rates}
    \begin{equation}
        \Gamma_\LI
        :=
        \min \left \{
            \Re \gamma_{\bm n} - R_{\bm n}(\mathbfcal L):
            \bm n \in \LJ
        \right\},
        \qquad
        \Gamma_\LI'
        :=
        \min \left \{
            \Re \gamma_{\bm n} :
            \bm n \in \LJ
        \right\}
        ,
    \end{equation}
    and an \emph{upper decay rate}
    \begin{equation}
        \begin{aligned}
            \Upsilon_\LI
            &:=
            \min \left \{
                \Re \gamma_{\bm n} :
                \bm n \in \LJ,~
                \Re \gamma_{\bm n} > \max_{\bm m \in \LI} \Re \gamma_{\bm m}
            \right\}.
        \end{aligned}
    \end{equation}
\end{definition}
\begin{proposition}
    For any truncation of the hierarchy~$\LI$, the
    resolvents of the tail block satisfy
    \begin{equation}
        \begin{aligned}
            \left\|
            \mathbfcal R\left(z; \mathbfcal L_\LJJ\right)
            \right\|_1
            &\leqslant
            \frac{1}{\Re z + \Gamma_\LI},
            \qquad
            \Re z > -\Gamma_\LI, \\
            \left\|
            \mathbfcal R\left(z; \mathbfcal L_\LJJ'\right)
            \right\|_1
            &\leqslant
            \frac{1}{\Re z + \Gamma_\LI'},
            \qquad
            \Re z > -\Gamma_\LI',
        \end{aligned}
    \end{equation}
    where \(\mathbfcal L'\) is composed of the diagonal entries of the Liouvillian,
    $
        \mathcal L'_{\bm n\bm m}
        :=
        \mathcal L_{\bm n\bm m}\delta_{\bm n\bm m}.
    $

    \label{prop:simplest-bound}
\end{proposition}
\begin{proof}
    For \(\bm n \in \LJ\), one has
    $
        R_{\bm n}(\mathbfcal L_\LJJ)
        \leqslant
        R_{\bm n}(\mathbfcal L)
    $,
    since \(\mathbfcal L_\LJJ\) is a block of \(\mathbfcal L\). Lemma~\ref{lemma:inv-diagonal} provides a bound
    \begin{equation}
        \beta\left(z; \mathcal L_{\bm{nn}}\right) \geqslant 
        \dist\left(z; S\left(-\gamma_{\bm n}; \hat H\right)\right) \geqslant
        |\Re(z + \gamma_{\bm n})|.
        \label{eq:tmp}
    \end{equation}
    Thus, for each $z$ such that $\Re z > -\Gamma_\LI$ and $\bm n \in \LJ$, we have
    \begin{equation}
        \beta\left(z; \mathcal L_{\bm{nn}}\right) - R_{\bm n}(\mathbfcal L_\LJJ) \geqslant 
        \Re(z + \gamma_{\bm n}) - R_{\bm n}(\mathbfcal L) > 0,
    \end{equation}
    hence $z \notin G_{\bm n}(\mathbfcal L_\LJJ)$. Since $R_{\bm n}(\mathbfcal L) \leqslant \sqrt{C(\Re \gamma_{\bm n} + \gamma)} = o(\Re \gamma_{\bm n})$,
    \begin{equation}
        \inf_{\bm n \in \LJ} \left [ \beta\left(z; \mathcal L_{\bm{nn}}\right) - R_{\bm n}(\mathbfcal L_\LJJ) \right ] > 0,
    \end{equation}
    and $z \in \Omega(\mathbfcal L_\LJJ)$. Corollary~\ref{cor:resolvent-set} applied to $\mathbfcal L_\LJJ$ and bound~\eqref{eq:tmp} prove the proposition for the resolvent of $\mathbfcal L_\LJJ$ on the half-plane $\Re z > -\Gamma_\LI$.

    Since \(\mathbfcal L_\LJJ'\) has no off-diagonal hierarchy blocks, its Gershgorin radii vanish. Otherwise,
    the statement for \(\mathbfcal L_\LJJ'\) is proved in the same way.
\end{proof}

This Proposition clarifies the meaning of the lower decay rates: $\Gamma_{\LI}$ provides a simple bound on the domain and on the norm of the resolvent of the tail block $\mathbfcal L_\LJJ$, while $\Gamma_\LI'$ provides the same for the diagonal component of the block $\mathbfcal L_\LJJ$. Using Lemma~\ref{lemma:gershgorin-radii}, we can obtain convenient estimates for off-diagonal blocks of $\mathbfcal L$ through the upper decay rate:
\begin{equation}
    \begin{aligned}
        \left \| \mathbfcal L_\LJI \right \|_1
        &\leqslant
        \sup_{\bm n \in \LI} R_{\bm n}(\mathbfcal L)
        \leqslant
        \sqrt{C \left[\Upsilon_\LI + \gamma\right]}, \\
        \left \| \mathbfcal L_\LIJ \right \|_1
        &\leqslant
        \sup_{\bm n \in \LJ}
        \left[
            R_{\bm n}(\mathbfcal L) - R_{\bm n}\left(\mathbfcal L_\LJJ\right)
        \right] \\
        &\leqslant
        \max_{\bm n \in \LI}
        \max_{j \in \{1, \ldots, N\}}
        R_{\bm n + \bm e_j}(\mathbfcal L)
        \leqslant
        \sqrt{C\left[\Upsilon_\LI + \tilde \gamma\right]},
    \end{aligned}
    \label{eq:off-diag-block-bounds}
\end{equation}
where
$\tilde \gamma
:=
\gamma + \max_{j \in \{1, \ldots, N\}} \Re \gamma_j$.

Here, we introduce an important definition, which would further allow us to rigorously establish convergence of finite-dimensional approximations.
\begin{definition}
    We call a sequence of truncations $\left\{\LI_{k}\right\}_{k=1}^\infty$ \emph{exhausting}
    if for every finite $\mathbb I \subset \mathbb N_0^N$ there exists $k_0$ such that for each $k > k_0$, $\mathbb I \subset \LI_k$, and
    \begin{equation}
        \Lambda
        :=
        \limsup_{k \to \infty}
        \frac{\Upsilon_{\LI_k}}{\Gamma_{\LI_k}'}
        < +\infty.
    \end{equation}
\end{definition}
In simpler terms, any exhausting sequence of truncations eventually includes all the multiindices and the ratio of the upper and the lower decay rates is essentially bounded. The latter condition is needed so that norm of the off-diagonal Liouvillian blocks~\ref{eq:off-diag-block-bounds} does not grow too fast for the exhausting sequence of truncations.

Lower decay rates $\Gamma_\LI$ and $\Gamma_\LI'$ are different and conveniently used in different bounds, but they are essentially equivalent for any exhausting sequence of truncations.
\begin{lemma}
    For any exhausting sequence of truncations \(\{\LI_k\}_{k=1}^\infty\), $\Gamma'_{\LI_k} \to +\infty$ and $\Gamma_{\LI_k} / \Gamma'_{\LI_k} \to 1$ as $k \to \infty$.
    \label{lemma:decay-ratio}
\end{lemma}
The proof of this lemma is given in Appendix~\ref{sec:lemma-decay-ratio-proof}.

\begin{corollary}
    A bounded set \(\Omega \subset \mathbb C\) contains at most finitely many spectral values of \(\mathbfcal L\).
    \label{lemma:spectrum-discreteness}
\end{corollary}
\begin{proof}
    Choose an exhausting sequence of truncations. According to Lemma~\ref{lemma:decay-ratio}, $\Gamma_\LI \to +\infty$ along this sequence, so we can pick a truncation \(\LI\) such that
       $ -\Gamma_\LI < \inf_{z \in \Omega} \Re z$, 
    hence
    $
        \overline{\Omega}
        \subset
        \Omega\left(\mathbfcal L_\LJJ\right).
    $
    Therefore, \(\mathbfcal S_\LI(z)\) is holomorphic in a neighborhood of \(\overline{\Omega}\). Since \(\LI\) is finite, \(\det \mathbfcal S_\LI(z)\) is a scalar holomorphic function. Its zeros in the compact set \(\overline{\Omega}\) are isolated and hence finite in number~\cite{Ablowitz2003}. Since $\mathbfcal S_\LI(z)$ is the Schur complement of $z - \mathbfcal L$, these zeros are precisely the spectral values of \(\mathbfcal L\).
\end{proof}

\section{Finite-dimensional approximation}
\label{sec:truncation}

In this section, we fix a truncation scheme for the HEOM Liouvillian and study the spectral properties of the resulting finite-dimensional approximations.
\begin{definition}[Truncated Liouvillian]
    For a given truncation \(\LI\) and Liouvillian $\mathbfcal L$, we define the \emph{truncated Liouvillian} as
\begin{equation}
    \mathbfcal L_\LI
    :=
    \mathbfcal L_\LII
    -
    \mathbfcal L_\LIJ
    \left(\mathbfcal L'_\LJJ\right)^{-1}
    \mathbfcal L_\LJI.
\end{equation}
    \label{def:truncated-liouvillian}
\end{definition}

We keep a finite part of the hierarchy explicitly and account for the effect of the discarded tail through a Schur-complement-type terminator.
Eigenvalues of the approximation are given by the solutions of $\det\left(z-\mathbfcal L_\LI\right)=0$. Since \(\mathbfcal L_\LII\) and \(\mathbfcal L_\LIJ\) have a common left zero eigenvector, the truncated Liouvillian \(\mathbfcal L_\LI\) has the eigenvalue \(\lambda = 0\). The corresponding right eigenvector represents the steady state of the approximation.
In the following, we use the results of Sec.~\ref{sec:spectral} to establish convergence of the spectrum of \(\mathbfcal L_\LI\) to the spectrum of the exact Liouvillian for any exhausting sequence of truncations. We then discuss the implications for unstable modes.

\subsection{Spectral convergence}

\begin{lemma}
    For a truncation $\LI$ such that $\Gamma_\LI > 0$ and a complex $z$ such that $\Re z>-\Gamma_\LI$, let
    \begin{equation}
        \varepsilon_\LI(z)
        :=
        \frac{
            C\sqrt{(\Upsilon_\LI+\gamma)(\Upsilon_\LI+\tilde\gamma)}
        }{\Gamma_\LI}
        \left[
            \frac{|z|}{\Re z+\Gamma_\LI}
            +
            \frac{\sqrt{C(\Gamma_\LI'+\gamma)}}{\Gamma_\LI'}
        \right].
        \label{eq:schur-error-majorant}
    \end{equation}
    Then
    \begin{equation}
        \left \| \mathbfcal S_\LI(z) - z + {\mathbfcal L}_\LI \right \|_1
        \leqslant
        \varepsilon_\LI(z).
        \label{eq:schur-error-bound}
    \end{equation}
    Moreover, for any bounded $\Omega \subset \mathbb C$ and any exhausting sequence of truncations $\{\LI_k\}_{k=1}^{\infty}$,
    \begin{equation}
        \lim_{k\to\infty}
        \sup_{z\in\Omega}
        \varepsilon_{\LI_k}(z)
        =0.
    \end{equation}
    \label{lemma:schur-convergence}
\end{lemma}
The statement of this lemma follows from bounds on the resolvent of discarded blocks formulated in Prop.~\ref{prop:simplest-bound} and Lemma~\ref{lemma:decay-ratio}. The detailed proof is given in Appendix~\ref{sec:lemma-schur-convergence-proof}. Finally, we are ready to formulate one of the main results of the paper: eigenvalues of truncated Liouvillians converge to spectral values of the exact Liouvillian of HEOM.

\begin{theorem}[Spectral convergence]
Let \(K \subset \mathbb C\) be a compact set. For each \(\lambda \in \sigma_K(\mathbfcal L) := \sigma(\mathbfcal L) \cap K\), choose a bounded open neighborhood \(V_\lambda\) with rectifiable boundary \(\partial V_\lambda\) such that
    $\overline V_\lambda \cap \sigma(\mathbfcal L) = \{\lambda\}$,
and let
    $K_- := K \setminus \bigcup_{\lambda \in \sigma_K(\mathbfcal L)} V_\lambda$.
    Then for any exhausting sequence of truncations $\{\LI_k\}_{k=1}^\infty$ there is $k_0$ such that for all $k > k_0$ the following statements hold:
\begin{enumerate}
    \item
        Each neighborhood \(V_\lambda\) contains at least one eigenvalue of \(\mathbfcal L_{\LI_k}\). Moreover, the number of eigenvalues of \(\mathbfcal L_{\LI_k}\) in \(V_\lambda\), counted with algebraic multiplicity, coincides with the algebraic multiplicity of \(\lambda\) as an eigenvalue of \(\mathbfcal L\).

    \item
        The operator \(\mathbfcal L_{\LI_k}\) has no eigenvalues in \(K_-\).
\end{enumerate}
\label{th:spectrum-convergence}
\end{theorem}

\begin{proof}
    Let
        $K_+
        :=
        K \cup \bigcup_{\lambda \in \sigma_K(\mathbfcal L)} \overline{V_\lambda}$.
    Since \(K\) is compact and \(\sigma_K(\mathbfcal L)\) is finite by Lemma~\ref{lemma:spectrum-discreteness}, the set \(K_+\) is compact.

    For each \(\lambda \in \sigma_K(\mathbfcal L)\), define
    \begin{equation}
        M_\lambda
        :=
        \max_{z \in \partial V_\lambda}
        \left \|
        \mathbfcal R(z;\mathbfcal L)
        \right \|_1 .
    \end{equation}
    This maximum exists because \(\partial V_\lambda\) is compact and disjoint from \(\sigma(\mathbfcal L)\). Likewise, define
    \begin{equation}
        M_-
        :=
        \max_{z \in K_-}
        \left \|
        \mathbfcal R(z;\mathbfcal L)
        \right \|_1 ,
    \end{equation}
    which is finite because \(K_-\) is a compact subset of \(\varrho(\mathbfcal L)\). Finally, set
    \begin{equation}
        M := \max\left(M_-, \max_{\lambda \in \sigma_K(\mathbfcal L)} M_\lambda\right).
    \end{equation}

    By Lemma~\ref{lemma:schur-convergence} there exists \(k_0\) such that for all $k > k_0$
    \begin{equation}
        \sup_{z \in K_+}
        \left \|
        \mathbfcal S_{\LI_k}(z) - z + \mathbfcal L_{\LI_k}
        \right \|_1
        <
        \frac{1}{M}.
        \label{eq:small-schur-error}
    \end{equation}

    We first prove the statement for a fixed \(\lambda \in \sigma(\mathbfcal L)\cap K\). For \(t \in [0,1]\), define
    \begin{equation}
        \mathbfcal F_{\LI_k}(z,t)
        :=
        t\left(z-\mathbfcal L_{\LI_k}\right)
        +
        (1-t)\mathbfcal S_{\LI_k}(z).
    \end{equation}
    Since
    \begin{equation}
        \mathbfcal F_{\LI_k}(z,t)
        =
        \left[
            \mathbfcal I
            -
            t\left(
                \mathbfcal S_{\LI_k}(z)-z+\mathbfcal L_{\LI_k}
            \right)
            \mathbfcal S_{\LI_k}(z)^{-1}
        \right]
    \mathbfcal S_{\LI_k}(z),
    \end{equation}
    it is enough to show that the bracket is invertible for all $t\in[0, 1]$ and all $z \in \partial V_\lambda$. By the Schur-complement representation,
        $\mathbfcal S_{\LI_k}(z)^{-1}
        =
        \mathbfcal R_{\LI_k \LI_k}(z;\mathbfcal L)
        $,
    and therefore
    \begin{equation}
        \left \|
        \mathbfcal S_{\LI_k}(z)^{-1}
        \right \|_1
        \leqslant
        \left \|
        \mathbfcal R(z;\mathbfcal L)
        \right \|_1
        \leqslant
        M
        \qquad
        \text{for } z \in \partial V_\lambda.
    \end{equation}
    Together with~\eqref{eq:small-schur-error}, this yields
    \begin{equation}
        \left \|
        t\left(
            \mathbfcal S_{\LI_k}(z)-z+\mathbfcal L_{\LI_k}
        \right)
        \mathbfcal S_{\LI_k}(z)^{-1}
        \right \|_1
        < 1,
    \end{equation}
    so \(\mathbfcal F_{\LI_k}(z,t)\) is invertible for all \(z \in \partial V_\lambda\) and all \(t \in [0,1]\).
    Since \(\LI_k\) is finite, \(\mathbfcal F_{\LI_k}(z,t)\) is matrix-valued.      Let us denote
    $f_t(z) := \det \mathbfcal F_{\LI_k}(z,t)$.
    For each fixed \(t\), the function \(f_t(z)\) is holomorphic in \(V_\lambda\) and continuous on \(\partial V_\lambda\). Since $\mathbfcal F_{\LI_k}(z, t)$ is invertible,
    $f_t(z) \neq 0$
    for all  $z \in \partial V_\lambda$, $t \in [0,1]$.
    Hence the logarithmic derivative \(f_t'(z)/f_t(z)\) is holomorphic in a neighborhood of \(\partial V_\lambda\), and the argument principle~\cite{Ablowitz2003} gives
    \begin{equation}
        N_\lambda(t)
        :=
        \frac{1}{2\pi i}
        \oint_{\partial V_\lambda}
        \frac{f_t'(z)}{f_t(z)}\,\mathrm dz,
        \label{eq:Nt-argument-principle}
    \end{equation}
    where \(N_\lambda(t)\) is exactly the number of zeros of \(f_t\) in \(V_\lambda\), counted with multiplicity. For \(t=0\), $N_\lambda(t)$ is equal to the number of the spectral values of \(\mathbfcal L\) in \(V_\lambda\). For \(t=1\), $N_\lambda(t)$ gives the number of the eigenvalues of \(\mathbfcal L_\LI\) in \(V_\lambda\).

    We show that \(N_\lambda(t)\) is continuous in \(t\). Since \(\partial V_\lambda \times [0,1]\) is compact and \(f_t(z)\) is continuous in $(z,t)$, the function \(|f_t(z)|\) attains a strictly positive minimum on \(\partial V_\lambda \times [0,1]\), therefore $1 / f_t(z)$ is also continuous in $(z, t)$. Moreover, because \(\mathbfcal F_{\LI_k}(z,t)\) is continuous in \((z,t)\) on \(\partial V_\lambda \times [0,1]\) and holomorphic in \(z\), the derivative \(f_t'(z)\) depends continuously on \((z,t)\) on \(\partial V_\lambda \times [0,1]\). Hence, the logarithmic derivative ${f_t'(z)}/{f_t(z)}$ is continuous on \(\partial V_\lambda \times [0,1]\). It follows that the integral in~\eqref{eq:Nt-argument-principle} depends continuously on \(t\), so \(N_\lambda(t)\) is a continuous function on \([0,1]\). On the other hand, \(N_\lambda(t)\) is integer-valued by the argument principle. Being both continuous and integer-valued on the connected set \([0,1]\), it must be constant, therefore $N_\lambda(0) = N_\lambda(1)$.

    We now prove the second statement. For \(z \in K_-\), we can express
    \begin{equation}
        z-\mathbfcal L_{\LI_k}
        =
        \left[
            \mathbfcal I
            -
            \left(
                \mathbfcal S_{\LI_k}(z)-z+\mathbfcal L_{\LI_k}
            \right)
            \mathbfcal S_{\LI_k}(z)^{-1}
        \right]
    \mathbfcal S_{\LI_k}(z).
    \end{equation}
    Since \(K_- \subset \varrho(\mathbfcal L)\), we have
    \begin{equation}
        \left \|
        \mathbfcal S_{\LI_k}(z)^{-1}
        \right \|_1
        \leqslant
        \left \|
        \mathbfcal R(z;\mathbfcal L)
        \right \|_1
        \leqslant
        M
        \qquad
        \text{for all } z \in K_-.
    \end{equation}
    Hence, by~\eqref{eq:small-schur-error},
    \begin{equation}
        \left \|
        \left(
            \mathbfcal S_{\LI_k}(z)-z+\mathbfcal L_{\LI_k}
        \right)
        \mathbfcal S_{\LI_k}(z)^{-1}
        \right \|_1
        < 1
        \qquad
        \text{for all } z \in K_-.
    \end{equation}
    Therefore, \(z-\mathbfcal L_{\LI_k}\) is invertible for all \(z \in K_-\), and \(\mathbfcal L_{\LI_k}\) has no eigenvalues in \(K_-\).
\end{proof}

\subsection{Stability of the truncated Liouvillian}
In order to analyze stability of the truncated Liouvillian, we are interested in its eigenvalues with non-negative real part. We are going to use Theorem~\ref{th:spectrum-convergence} to show that unstable eigenvalues of the truncated Liouvillian approach the spectral values of the exact one. But to apply it, we need to locate the relevant eigenvalues within a compact set.
\begin{lemma}
    Let $\sigma_+(\mathbfcal A) := \{\lambda \in \sigma(\mathbfcal A): \Re \lambda \geqslant 0\}$ for an arbitrary operator $\mathbfcal A$.
    For any exhausting sequence of truncations $\{\LI_k\}_{k=1}^\infty$, there exist a compact set $G_+(\mathbfcal L)$ and $k_0$ such that $\sigma_+(\mathbfcal L) \subset G_+(\mathbfcal L)$ and, for every $k>k_0$,
    $\sigma_+(\mathbfcal L_{\LI_k}) \subset G_+(\mathbfcal L)$.
    \label{lemma:unstable-spectrum}
\end{lemma}
The proof is based on Theorem~\ref{th:gershgorin-hard} which allows us to locate the eigenvalues of the approximation, bounds given in Lemmas~\ref{lemma:inv-diagonal} and~\ref{lemma:gershgorin-radii}, and bounds on the Liouvillian blocks~\eqref{eq:off-diag-block-bounds}. For clarity, we give the full proof in Appendix~\ref{sec:lemma-unstable-spectrum-proof}. Now, we proceed to another main result of our work. We prove that sufficiently large truncations do not create spurious instabilities.

\begin{definition}
    We call a Liouvillian \(\mathbfcal L\) \emph{stable} if it has no eigenvalues with positive real part. We call \(\mathbfcal L\) \emph{gapped} if \(0\) is its only eigenvalue with vanishing real part, counted with algebraic multiplicity.
\end{definition}

\begin{theorem}[Stability]
    Let the Liouvillian \(\mathbfcal L\) be stable and gapped. Then, for any exhausting sequence of truncations $\{\LI_k\}_{k=1}^{\infty}$ there is $k_0$ such that for all $k > k_0$ the truncated Liouvillians \(\mathbfcal L_{\LI_k}\) are stable and gapped.
    \label{th:stability}
\end{theorem}

\begin{proof}
    By construction, \(\mathbfcal L_\LI\) has the eigenvalue \(\lambda = 0\).

    Let \(V_0\) be a bounded open neighborhood of \(\lambda = 0\) such that
    $\overline V_0 \cap \sigma(\mathbfcal L) = \{0\}$.
    By Lemma~\ref{lemma:unstable-spectrum}, there exists a compact set \(G_+(\mathbfcal L)\) containing all eigenvalues of \(\mathbfcal L\) and all sufficiently large \(\mathbfcal L_{\LI_k}\) with non-negative real part. Applying Theorem~\ref{th:spectrum-convergence} to the compact set \(G_+(\mathbfcal L)\), we conclude that there exists $k_0$ such that for every $k > k_0$
    \begin{enumerate}
        \item
            the neighborhood \(V_0\) contains exactly one eigenvalue of \(\mathbfcal L_{\LI_k}\), counted with algebraic multiplicity;

        \item
            there are no eigenvalues of \(\mathbfcal L_{\LI_k}\) in
        $G_+(\mathbfcal L) \setminus V_0$.
    \end{enumerate}
    Since all eigenvalues of \(\mathbfcal L_{\LI_k}\) with non-negative real part belong to \(G_+(\mathbfcal L)\), it follows that \(0\) is the only eigenvalue of \(\mathbfcal L_{\LI_k}\) with non-negative real part, for all $k > k_0$. Therefore, all \(\mathbfcal L_{\LI_k}\) are stable and gapped.
\end{proof}

\section{Discussion and example}
\label{sec:discussion}

In the proofs of Theorems~\ref{th:spectrum-convergence} and~\ref{th:stability}, we rely on bounds for the \emph{exact resolvent} \(\mathbfcal R(z; \mathbfcal L)\) which is not directly accessible. In particular, we do not provide any easily verifiable criterion that would determine whether a given truncation $\LI$ is sufficiently deep. If the truncated Liouvillian \(\mathbfcal L_\LI\) has unstable eigenvalues, our results do not determine whether these are spurious truncation artifacts or whether the exact Liouvillian \(\mathbfcal L\) has unstable spectral values. Also, the explicit bounds~\eqref{eq:schur-error-bound} based on Gershgorin's theorem turn out to be very rough. In practice, the eigenvalues may converge much earlier than the presented formalism can even be applied.

Our analysis relies strongly on the symmetric $\sim \sqrt{n}$ scaling of the couplings between the hierarchy levels. HEOM also has seemingly equivalent formulations in which the lowering and raising off-diagonal blocks of the Liouvillian have different growth rates~\cite{Shi2009, Zhu2011, Chen2024}. Equation~\eqref{eq:HEOM0} is formally invariant under the scaling transformation
\begin{equation}
    \hat{\tilde \rho}_{\bm n} := f_{\bm n} \hat \rho_{\bm n},\quad
    \tilde{\mathcal L}_{\bm n \bm m} := \frac{f_{\bm n}}{f_{\bm m}} \mathcal L_{\bm n \bm m},
\end{equation}
where each $f_{\bm n}$ is nonzero. However, the hierarchy $\|\cdot\|_1$ norms before and after this transformation are equivalent if and only if exist $f_\mathrm u > f_\mathrm l > 0$ such that $f_\mathrm u > |f_{\bm n}| > f_\mathrm l$ holds for all $\bm n$. Otherwise, the transformation changes the topology of the ADO hierarchy space, unless one additionally adjusts the norm of the scaled ADOs. The spectrum of an unbounded operator usually depends on the chosen Banach space and the operator domain.

The arguments based on Theorem~\ref{th:gershgorin-hard} remain valid when the off-diagonal blocks $\mathcal L_{\bm n,\bm n\pm\bm e_j}$ grow slower than linearly with the hierarchy level $n_j$, after adjusting the bounds on the Gershgorin radii~\eqref{eq:bound-off-diag} and the Schur-complement approximation~\eqref{eq:schur-error-bound}. In other formulations of HEOM~\cite{Ishizaki2005, Tanimura2006}, the couplings $\mathcal L_{\bm n,\bm n-\bm e_j}$ grow linearly with $n_j$, so the presented reasoning is no longer valid. Even discreteness of the Liouvillian spectrum is no longer guaranteed. Despite the formal equivalence of HEOM variants related by the topology changing similarity transformations, they may have different functional-analytic properties. This, in turn, may affect behavior of the practical numerical schemes.

\begin{table}[t]
    \caption{Poles and residues of the approximation~\eqref{eq:coth} for $T = 0.5$ and $\Lambda = 50$.}
    \begin{center}
        \begin{tabular}{c|c|c|c|c|c|c|c|c|c|c|c}
            \hline\hline
            $j$ & 1 & 2 & 3 & 4 & 5 & 6 & 7 & 8 & 9 & 10 & 11 \\
            \hline
            $\nu_{j}$ &
            341.9 &
            110.8 &
            63.00 &
            41.63 &
            29.38 &
            21.57 &
            16.35 &
            12.63 &
            9.426 &
            6.283 &
            3.145 \\
            \hline
            $r_j$ &
            219.2 &
            24.98 &
            9.358 &
            4.958 &
            3.055 &
            2.009 &
            1.367 &
            1.058 &
            1.002 &
            1.000 &
            1.000 \\
            \hline\hline
        \end{tabular}
    \end{center}
    \label{tab:coth}
\end{table}

\begin{figure}[t]
    \includegraphics[width = \linewidth]{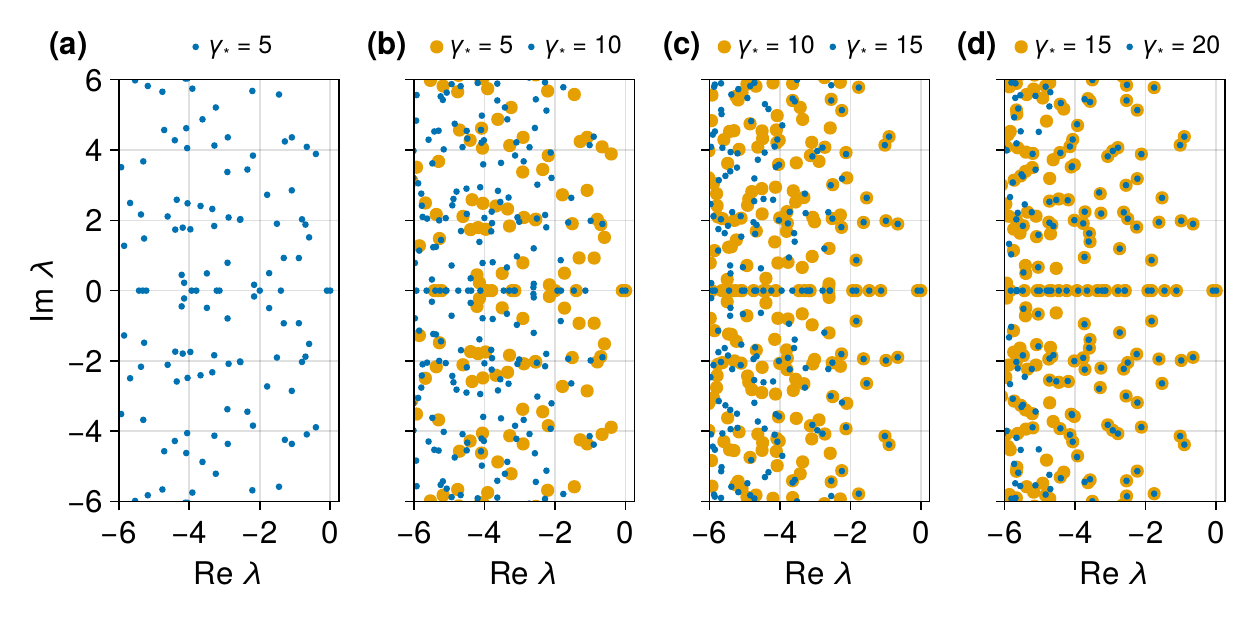}
    \caption{Eigenvalues of truncated Liouvillians $\mathbfcal L_{\LI(\gamma_\ast)}$. Here, $\alpha = 2$, $\omega_0 = 2$, and $\eta = 0.5$. Poles and residues of the $\coth$ approximation are given in Tab.~\ref{tab:coth}.}
    \label{fig:spectrum}
\end{figure}

To complement the general analysis, we consider a simple example which illustrates convergence of the spectral values. We use a paradigmatic spin-boson model coupled to a damped harmonic oscillator. The Hamiltonian and coupling operator are
\begin{equation}
    \hat H = \hat \sigma_x,\qquad \hat q = \hat \sigma_z,
\end{equation}
where \(\hat \sigma_x\) and \(\hat \sigma_z\) are Pauli matrices. We choose the spectral density in the form
\begin{equation}
    J(\omega) = \frac{\alpha \omega \omega_0^4}{\left(\omega^2 - \omega_0^2\right)^2 + 4 \eta^2 \omega^2},
\end{equation}
where \(\alpha\) controls the coupling strength, while \(\omega_0\) and \(\eta\) are the frequency and damping rate of the harmonic oscillator, respectively. Following the approach of Ref.~\cite{Vadimov2025}, we separate the dynamical poles and the fluctuation poles, which originate from the poles of the spectral density and the Bose--Einstein distribution function, respectively. We then construct a rational approximation of
\begin{equation}
    \coth\left(\frac{\omega}{2 T}\right) \approx \operatorname{bose}_T(\omega) := \frac{2 T}{\omega} + \sum\limits_{j = 1}^{N - 2} \frac{2 r_j \omega}{\omega^2 + \nu_{j}^2}
    \label{eq:coth}
\end{equation}
using the symmetrized AAA algorithm~\cite{Nakatsukasa2018} on a sufficiently large interval \([-\Lambda, \Lambda]\). The HEOM decay rates are then given by
\begin{equation}
    \begin{aligned}
        \gamma_1 := \eta + i \sqrt{\omega_0^2 - \eta^2},\quad
        \gamma_2 := \eta - i \sqrt{\omega_0^2 - \eta^2},\quad
        \gamma_j := \nu_{j-2},\quad j = 3, \ldots, N.
    \end{aligned}
\end{equation}
The coupling constants are given by
\begin{equation}
    \begin{gathered}
        c_1 := \frac{i \omega_0^2 \sqrt{\alpha}}{2 \sqrt{2 \eta} \left[\omega_0^2 - \eta^2\right]^{1/4}}, \quad
        c_1' := \frac{c_1}{2}\left[1 + \operatorname{bose}_T(-i \gamma_1)\right], \quad
        c_1'' := \frac{c_1}{2}\left[1 - \operatorname{bose}_T(-i \gamma_1)\right],\\
        c_2 := c_1, \quad 
        c_2' := -\frac{c_1}{2}\left[1 + \operatorname{bose}_T(-i \gamma_2)\right], \quad
        c_2'' := -\frac{c_1}{2}\left[1 - \operatorname{bose}_T(-i \gamma_2)\right],\\
        c_j := i \sqrt{\frac{r_{j-2} \gamma_j \omega_0^4}{2 \left[(\gamma_j^2 + \omega_0^2)^2 - \gamma_j^2 \eta^2\right]}} \quad
        c_j' := -c_j, \quad c_j'' := c_j,\quad j = 3, \ldots, N.
    \end{gathered}
\end{equation}

For the analysis of the truncated HEOM, we parametrize the truncations by the decay rate
\begin{equation}
    \LI(\gamma_\ast) := \left \{\bm n \in \mathbb N_0^N: \Re \gamma_{\bm n} \leqslant \gamma_\ast\right \},
\end{equation}
where $\gamma_\ast$ determines the accuracy of the truncation. For any growing sequence of $\gamma_\ast^{(n)}$, the respective sequence of truncations is exhausting because $\Gamma'_{\LI(\gamma_\ast)} > \gamma_\ast$ and $\Upsilon_{\LI(\gamma_\ast)} = \Gamma'_{\LI(\gamma_\ast)}$.
In Fig.~\ref{fig:spectrum}, we plot the eigenvalues of \(\mathbfcal L_{\LI(\gamma_\ast)}\) for four different values of \(\gamma_\ast\). As \(\gamma_\ast\) increases, the eigenvalues of the Liouvillian with the largest real parts begin to converge. However, one cannot even employ the presented formalism for this example since all $\Gamma_{\LI(\gamma_\ast)}$ turn out to be negative for each $\gamma_\ast = 5, 10, 15$, and $20$.

\section{Conclusions}
\label{sec:conclusions}

In conclusion, we have shown that for finite-dimensional open quantum systems, a suitable truncation of HEOM is free of spectral pollution or spurious instabilities. Our analysis relies on the partial block-diagonal dominance of the HEOM Liouvillian, which yields convenient Gershgorin-type bounds for its resolvent. These bounds, though, turn out quite pessimistic, which limits the practical implications of our results. Still, we believe that this work represents a valuable contribution to the mathematical foundations of computational quantum physics. 

In this work, we do not address the intrinsic stability of the exact HEOM. A recently established equivalence between HEOM and a pseudomode Lindblad master equation~\cite{Mueller2026} implies contractivity in the trace norm of the joint system--pseudomode density operator. However, the topology induced by this norm may differ substantially from that induced by the hierarchy norm $\|\cdot\|_1$ used here, so this result does not directly imply stability in the present setting. A rigorous analysis of HEOM stability, together with extensions to infinite-dimensional and driven open quantum systems, remains an important direction for future work.

\section*{Acknowledgements}
We thank Malte Krug, J\"urgen Stockburger, and Paolo Muratore-Ginanneschi for fruitful discussions.
This research was funded by the Research Council of Finland Centre of Excellence program (project Nos. 352925 and 336810) and grant No. 349594 (THEPOW), and by the European Research Council under Advanced Grant No. 101053801 (ConceptQ).

\appendix
\numberwithin{equation}{section}

\section{Proofs of spectral theorems and lemmas}

\subsection{Proof of Theorem~\ref{th:gershgorin-hard}}
\label{sec:th-gershgorin-hard-proof}
\begin{proof}
We split the Liouvillian into the diagonal and off-diagonal parts with respect to the hierarchy multiindex,
$\mathbfcal L = \mathbfcal L' + \mathbfcal L''$,
where
\begin{equation}
    \mathcal L'_{\bm n\bm m}
    =
    \mathcal L_{\bm n\bm m}\delta_{\bm n\bm m},
    \qquad
    \mathcal L''_{\bm n\bm m}
    =
    \mathcal L_{\bm n\bm m}(1-\delta_{\bm n\bm m}).
\end{equation}
For fixed \(z\), introduce the diagonal operator
\begin{equation}
    \mathbfcal D(z) := z - \mathbfcal L'.
\end{equation}
    Since \(z \in \Omega(\mathbfcal L)\), we have
\begin{equation}
    \beta(z;\mathcal L_{\bm{nn}}) - R_{\bm n}(\mathbfcal L) > 0
    \qquad
    \text{for all }\bm n.
    \label{eq:beta-greater}
\end{equation}
    In particular, \(\beta(z;\mathcal L_{\bm{nn}})>0\), hence \(z \notin \sigma(\mathcal L_{\bm n\bm n})\) for all \(\bm n\), and each block $z-\mathcal L_{\bm n\bm n}$
is invertible. Since these blocks act on a finite-dimensional space, one has
\begin{equation}
    \left\|
    (z-\mathcal L_{\bm n\bm n})^{-1}
    \right\|
    =
    \beta(z; \mathcal L_{\bm{nn}})^{-1}.
    \label{eq:block-inverse-beta}
\end{equation}
Therefore \(\mathbfcal D(z)\) is invertible, with
\begin{equation}
    \left \|\mathbfcal D(z)^{-1} \right \|_1
    =
    \sup_{\bm n}\beta(z; \mathcal L_{\bm{nn}})^{-1}.
\end{equation}

Now consider the operator
\begin{equation}
    \mathbfcal K(z)
    :=
    \mathbfcal L''\,\mathbfcal D(z)^{-1}.
\end{equation}
For arbitrary \(\hat{\bm\rho}\),
\begin{equation}
    \begin{aligned}
        \left \|\mathbfcal K(z)\hat{\bm\rho}\right \|_1
        &=
        \sum_{\bm n}
        \left\|
        \sum_{\bm m\ne \bm n}
        \mathcal L_{\bm n\bm m}
        (z-\mathcal L_{\bm m\bm m})^{-1}
        \hat\rho_{\bm m}
        \right\|
        \\
        &\leqslant
        \sum_{\bm n}\sum_{\bm m\ne \bm n}
        \|\mathcal L_{\bm n\bm m}\|\,
        \left \|(z-\mathcal L_{\bm m\bm m})^{-1}\right \|\,
        \|\hat\rho_{\bm m}\|
        \\
        &=
        \sum_{\bm m}
        \left \|\hat\rho_{\bm m}\right \|
        \sum_{\bm n\ne \bm m}
        \|\mathcal L_{\bm n\bm m}\|\,
        \left \|(z-\mathcal L_{\bm m\bm m})^{-1}\right \|
        \\
        &\leqslant
        \|\hat{\bm\rho}\|_1
        \sup_{\bm m}
        \frac{R_{\bm m}(\mathbfcal L)}
        {\beta(z;\mathcal L_{\bm{mm}})}
        =
        q(z;\mathbfcal L)\left\|\hat{\bm\rho}\right\|_1 < \left\|\hat{\bm\rho}\right\|_1.
    \end{aligned}
\end{equation}
Thus, $\|\mathbfcal K(z)\|_1 \leqslant q(z;\mathbfcal L) < 1$ and
 \(\mathbfcal I-\mathbfcal K(z)\) is invertible by the Neumann series, and
\begin{equation}
    \left \|[\mathbfcal I-\mathbfcal K(z)]^{-1}\right \|_1
    \leqslant
    \frac{1}{1-q(z;\mathbfcal L)}.
\end{equation}
Finally, on the natural domain of \(\mathbfcal L\),
\begin{equation}
    z-\mathbfcal L
    =
    z-\mathbfcal L' - \mathbfcal L''
    =
    \left[\mathbfcal I-\mathbfcal L''\mathbfcal D(z)^{-1}\right]\mathbfcal D(z)
    =
    [\mathbfcal I-\mathbfcal K(z)]\mathbfcal D(z).
\end{equation}
Hence, \(z-\mathbfcal L\) is invertible and
    $\mathbfcal R(z;\mathbfcal L)
    =
    \mathbfcal D(z)^{-1}
    [\mathbfcal I-\mathbfcal K(z)]^{-1}.
    $

It remains to prove the sharper bound \eqref{eq:resolvent-bound-hard-clean}. For any \(\hat{\bm\rho}\) in the domain of \(\mathbfcal L\),
\begin{equation}
    \begin{aligned}
        \left \|(z-\mathbfcal L)\hat{\bm\rho} \right \|_1
        &=
        \sum_{\bm n}
        \left\|
        (z-\mathcal L_{\bm n\bm n})\hat\rho_{\bm n}
        -
        \sum_{\bm m\ne \bm n}
        \mathcal L_{\bm n\bm m}\hat\rho_{\bm m}
        \right\|
        \\
        &\geqslant
        \sum_{\bm n}
        \left \|(z-\mathcal L_{\bm n\bm n})\hat\rho_{\bm n}\right \|
        -
        \sum_{\bm n}
        \left\|
        \sum_{\bm m\ne \bm n}
        \mathcal L_{\bm n\bm m}\hat\rho_{\bm m}
        \right\|
        \\
        &\geqslant
        \sum_{\bm n}
        \beta(z;\mathcal L_{\bm{nn}})\|\hat\rho_{\bm n}\|
        -
        \sum_{\bm n}\sum_{\bm m\ne \bm n}
        \|\mathcal L_{\bm n\bm m}\|\,\|\hat\rho_{\bm m}\|
        \\
        &=
        \sum_{\bm n}
        \left[
            \beta(z;\mathcal L_{\bm{nn}})
            -R_{\bm n}(\mathbfcal L)
        \right]
        \|\hat\rho_{\bm n}\|
        \\
        &\geqslant
        \inf_{\bm n}
        \left[
            \beta(z;\mathcal L_{\bm{nn}})
            -R_{\bm n}(\mathbfcal L)
        \right]
        \|\hat{\bm\rho}\|_1.
    \end{aligned}
\end{equation}
Positivity of the expressions in the second and third lines of the above equation is provided by Eq.~\eqref{eq:beta-greater}.
Therefore
\begin{equation}
    \|\mathbfcal R(z;\mathbfcal L)\|_1
    \leqslant
    \left(
    \inf_{\bm n}
    \left[
        \beta(z;\mathcal L_{\bm{nn}})
        -R_{\bm n}(\mathbfcal L)
    \right]
    \right)^{-1}
    =
    \sup_{\bm n}
    \left[
        \beta(z;\mathcal L_{\bm{nn}})
        -R_{\bm n}(\mathbfcal L)
    \right]^{-1}.
\end{equation}
\end{proof}

\subsection{Proof of Lemma~\ref{lemma:inv-diagonal}}
\label{sec:lemma-inv-diagonal-proof}

\begin{proof}
    Since $\mathcal T$ acts on a finite-dimensional space, its spectrum is given by
    \begin{equation}
        \sigma(\mathcal T) = \left\{
            b + i (\lambda_1 - \lambda_2) : \lambda_1, \lambda_2 \in \sigma\left(\hat A\right)
        \right\} \subset S\left(b, \hat A\right).
    \end{equation}
    Hence, any \(z \notin S\left(b,\hat A\right)\) belongs to the resolvent set $\varrho(\mathcal T)$. For an arbitrary real \(\varphi \in [0,2\pi)\), define
    \begin{equation}
        \hat X_\varphi(t)
        :=
        \exp\left[(\mathcal T - z) e^{i \varphi} t \right] \hat \rho = 
        \exp\left[
            -(z-b)e^{i\varphi} t
        \right]
        \exp\left(
            -i e^{i\varphi} \hat A t
        \right)
        \hat \rho
        \exp\left(
            i e^{i\varphi} \hat A t
        \right).
    \end{equation}
    For every \(\varphi\) such that
    \begin{equation}
        c_\varphi(z)
        :=
        \Re\left[(z-b)e^{i\varphi}\right]
        -
        \Delta(\hat A)\left|\sin\varphi\right|
        >
        0,
        \label{eq:cphi-positive}
    \end{equation}
    the operator \(\hat X_\varphi(t)\) decays to zero as \(t \to +\infty\), and therefore
    \begin{equation}
        (z-\mathcal T)^{-1} e^{-i\varphi} \hat \rho
        =
        \int\limits_0^{+\infty} \hat X_\varphi(t)\,\mathrm d t.
        \label{eq:resolvent-rotated-integral}
    \end{equation}

    Let us estimate the norm of the integrand. Since
    the trace norm is invariant under unitary left and right multiplication, we have
    \begin{equation}
        \begin{aligned}
            \left\|
            \exp\left(
                -i e^{i\varphi} \hat A t
            \right)
            \hat \rho
            \exp\left(
                i e^{i\varphi} \hat A t
            \right)
            \right\|
            &=
            \left\|
            \exp\left(
                \hat A t \sin\varphi
            \right)
            \hat \rho
            \exp\left(
                -\hat A t \sin\varphi
            \right)
            \right\| \\
            &\leqslant
            \left\|
            \exp\left(
                \hat A t \sin\varphi
            \right)
            \right\|
            \left\|
            \exp\left(
                -\hat A t \sin\varphi
            \right)
            \right\|
            \left\|
            \hat \rho
            \right\| \\
            &=
            \exp\left[
            \Delta(\hat A)\left|\sin\varphi\right| t
            \right] \left \| \hat \rho \right \|.
        \end{aligned}
    \end{equation}
    Thus,
    \begin{equation}
        \left\|
        \hat X_\varphi(t)
        \right\|
        \leqslant
        \exp\left[
            -c_\varphi(z)t
        \right]
        \left\|
        \hat \rho
        \right\|.
    \end{equation}
    Combining this with \eqref{eq:resolvent-rotated-integral}, we obtain
    \begin{equation}
        \left\|
        (z-\mathcal T)^{-1}\hat \rho
        \right\|
        \leqslant
        \frac{\left\|
        \hat \rho
        \right\|}{c_\varphi(z)}.
    \end{equation}
    Since this is valid for every \(\varphi\) satisfying \eqref{eq:cphi-positive},
    \begin{equation}
        \left\|
        (z-\mathcal T)^{-1}
        \right\|
        \leqslant
        \left[
            \max_{\varphi \in [0,2\pi)}
            c_\varphi(z)
        \right]^{-1}.
        \label{eq:resolvent-bound-cphi}
    \end{equation}

    It remains to compute the maximum. Let
    $x := \Re(z-b)$ and 
    $y := \Im(z-b)$.
    Then
    \begin{equation}
        c_\varphi(z)
        =
        x\cos\varphi - y\sin\varphi - \Delta(\hat A)\left|\sin\varphi\right|.
    \end{equation}

    If \(|y| \leqslant \Delta(\hat A)\), then
    \begin{equation}
        \begin{aligned}
            c_\varphi(z)
            &\leqslant
            |x|\,|\cos\varphi|
            +
            \left(
                |y|-\Delta(\hat A)
            \right)
            |\sin\varphi|
            \leqslant
            |x|.
        \end{aligned}
    \end{equation}
    Equality is attained at \(\varphi = 0\) if \(x \geqslant 0\) and at \(\varphi=\pi\) if \(x<0\). Hence
    \begin{equation}
        \max_{\varphi} c_\varphi(z)
        =
        |x|
        =
        \operatorname{dist}\left(
            z,
            S(b,\hat A)
        \right).
    \end{equation}

    If \(y > \Delta(\hat A)\), let \(c := y-\Delta(\hat A) > 0\). Then for every \(\varphi\),
    \begin{equation}
        c_\varphi(z)
        \leqslant
        x\cos\varphi - c\sin\varphi.
    \end{equation}
    Indeed, for \(\sin\varphi < 0\) this is an equality, while for \(\sin\varphi \geqslant 0\) we have
    \begin{equation}
        -y\sin\varphi - \Delta(\hat A)\left|\sin\varphi\right|
        =
        -(y+\Delta(\hat A))\sin\varphi
        \leqslant
        -c\sin\varphi.
    \end{equation}
    Therefore,
    \begin{equation}
        c_\varphi(z)
        \leqslant
        \sqrt{x^2 + c^2}.
    \end{equation}
    Equality is attained for
    \begin{equation}
        \cos\varphi = \frac{x}{\sqrt{x^2+c^2}},
        \qquad
        \sin\varphi = -\frac{c}{\sqrt{x^2+c^2}},
    \end{equation}
    so
    \begin{equation}
        \max_{\varphi} c_\varphi(z)
        =
        \sqrt{x^2 + \left(y-\Delta(\hat A)\right)^2}
        =
        \left|
            z - b - i\Delta(\hat A)
        \right|
        =
        \operatorname{dist}\left(
            z,
            S(b,\hat A)
        \right).
    \end{equation}

    If \(y < -\Delta(\hat A)\), the argument is analogous with
    \begin{equation}
        c := -y-\Delta(\hat A) > 0,
    \end{equation}
    and yields
    \begin{equation}
        \max_{\varphi} c_\varphi(z)
        =
        \sqrt{x^2 + \left(y+\Delta(\hat A)\right)^2}
        =
        \left|
            z - b + i\Delta(\hat A)
        \right|
        =
        \operatorname{dist}\left(
            z,
            S(b,\hat A)
        \right).
    \end{equation}

    Substituting this into \eqref{eq:resolvent-bound-cphi}, we obtain
    \begin{equation}
        \left\|
        (z-\mathcal T)^{-1}
        \right\|
        \leqslant
        \frac{1}{
            \operatorname{dist}\left(
                z,
                S(b,\hat A)
            \right)
        }.
    \end{equation}
\end{proof}

\subsection{Proof of Lemma~\ref{lemma:gershgorin-radii}}
\label{sec:lemma-gershgorin-radii-proof}
\begin{proof}
    For the off-diagonal blocks of the HEOM Liouvillian, we have
    \begin{equation}
        \begin{aligned}
            \left \| \mathcal L_{\bm n,\, \bm n + \bm e_j} \right \|
            \leqslant
            2 |c_j| \sqrt{n_j + 1}\left \| \hat q_j \right \|,~ 
            \left \| \mathcal L_{\bm n,\, \bm n - \bm e_j} \right \|
            \leqslant
            \sqrt{n_j}\left \| \hat q_j \right \|
            \left(
                |c_j'| + |c_j''|
            \right).
        \end{aligned}
    \end{equation}
    Therefore,
    \begin{equation}
        \begin{aligned}
            R_{\bm n}(\mathbfcal L)
            &\leqslant
            \sum\limits_{j=1}^N
            \left \| \hat q_j \right \|
            \left[
                2 |c_j| \sqrt{n_j}
                +
                \left(
                    |c_j'| + |c_j''|
                \right)\sqrt{n_j + 1}
            \right] \\
            &\leqslant
            \sum\limits_{j=1}^N
            \sqrt{n_j + 1}
            \left \| \hat q_j \right \|
            \left(
                2 |c_j| + |c_j'| + |c_j''|
            \right).
        \end{aligned}
    \end{equation}

    Denote $a_j
        :=
        \left \| \hat q_j \right \|
        \left(
            2 |c_j| + |c_j'| + |c_j''|
        \right)
    $ and $r_j := \Re \gamma_j > 0$. Then
    \begin{equation}
        R_{\bm n}(\mathbfcal L)
        \leqslant
        \sum\limits_{j=1}^N a_j \sqrt{n_j + 1}.
    \end{equation}

    We now maximize the right-hand side under the constraint
    \begin{equation}
        \sum\limits_{j=1}^N n_j r_j = \Gamma,
    \end{equation}
    where \(\Gamma = \Re \gamma_{\bm n}\). Since the function
    \(
        \sum_{j=1}^N a_j \sqrt{n_j+1}
    \)
    is concave in \((n_1,\dots,n_N)\), the maximum over \(n_j \geqslant -1\) is attained at a critical point of the corresponding Lagrange function
    \begin{equation}
        F(\bm n,\lambda)
        :=
        \sum\limits_{j=1}^N a_j \sqrt{n_j + 1}
        -
        \lambda
        \left(
            \sum\limits_{j=1}^N n_j r_j - \Gamma
        \right).
    \end{equation}
    The extremum conditions \(\partial F/\partial n_j = 0\) give
    \begin{equation}
        \frac{a_j}{2\sqrt{n_j+1}} = \lambda r_j,
    \end{equation}
    hence
    \begin{equation}
        n_j
        =
        \left(
            \frac{a_j}{2\lambda r_j}
        \right)^2 - 1.
    \end{equation}
    Substituting this into the constraint, we obtain
    \begin{equation}
        \begin{aligned}
            \Gamma
            =
            \sum\limits_{j=1}^N r_j
            \left[
                \left(
                    \frac{a_j}{2\lambda r_j}
                \right)^2 - 1
            \right]
            =
            \frac{1}{4\lambda^2}
            \sum\limits_{j=1}^N \frac{a_j^2}{r_j}
            -
            \sum\limits_{j=1}^N r_j.
        \end{aligned}
    \end{equation}
    Hence,
    \begin{equation}
        \lambda
        =
        \frac{1}{2}
        \sqrt{
            \frac{
                \sum\limits_{j=1}^N a_j^2/r_j
            }{
                \Gamma + \sum\limits_{j=1}^N r_j
            }
        }=
        \frac{1}{2}
        \sqrt{
            \frac{C}{\Gamma + \gamma}
        }.
    \end{equation}
    The last equality holds according to the definitions of \(a_j\), \(r_j\), \(C\), and \(\gamma\).
    Finally,
    \begin{equation}
        \begin{aligned}
            R_{\bm n}(\mathbfcal L)
            \leqslant
            \sum\limits_{j=1}^N a_j \sqrt{n_j+1}
            =
            \sum\limits_{j=1}^N
            \frac{a_j^2}{2\lambda r_j}
            =
            \frac{1}{2\lambda}
            \sum\limits_{j=1}^N \frac{a_j^2}{r_j}
            =
            \sqrt{
                C(\Gamma+\gamma)
            }=
            \sqrt{
                C(\Re \gamma_{\bm n} + \gamma)
            }.
        \end{aligned}
    \end{equation}
\end{proof}

\subsection{Proof of Corollary~\ref{cor:resolvent-set}}
\label{sec:corr-resolvent-set}
\begin{proof}
    For each $z \in \Omega(\mathbfcal L)$ and each $\bm n$
    \begin{equation}
        \frac{R_{\bm n}(\mathbfcal L)}{\beta(z; \mathcal L_{\bm{nn}})} < 1.
    \end{equation}
    On the other hand, due to Lemmas~\ref{lemma:inv-diagonal} and~\ref{lemma:gershgorin-radii}
    \begin{equation}
        \frac{R_{\bm n}(\mathbfcal L)}{\beta(z; \mathcal L_{\bm{nn}})} \leqslant 
        \frac{R_{\bm n}(\mathbfcal L)}{\operatorname{dist}\left(z, S\left(-\gamma_{\bm n}; \hat H\right)\right)} \leqslant
        \frac{\sqrt{C(\Re \gamma_{\bm n} + \gamma)}}{|\Re (z + \gamma_{\bm n})|}.
    \end{equation}

    If $\Re z > \gamma + C$, then for all $\bm n$
    \begin{equation}
        \frac{R_{\bm n}(\mathbfcal L)}{\beta(z; \mathcal L_{\bm{nn}})} < \frac{1}{2},
        \label{eq:one-half-bound}
    \end{equation}
    which results in $q(z; \mathbfcal L) \leqslant 1/2$.

    Otherwise, bound~\eqref{eq:one-half-bound} holds for all $\bm n$ such that
    \begin{equation}
        \Re \gamma_{\bm n} > 2C - \Re z + 2 \sqrt{C(C + \gamma - \Re z)}.
    \end{equation}
    Since $\Re \gamma_j > 0$ for each $j = 1,\ldots, N$, there are only finitely many multiindices $\bm n$ with bounded $\Re \gamma_{\bm n}$. Hence, only finitely many ratios $R_{\bm n}(\mathbfcal L) / \beta(z; \mathcal L_{\bm{nn}})$ can be greater than $1/2$. Each of them is still strictly less than $1$, therefore $q(z; \mathbfcal L) < 1$.

    Thus, for any $z \in \Omega(\mathbfcal L)$ condition~\eqref{eq:qz-condition} is satisfied, which proofs the claim.
\end{proof}

\subsection{Proof of Lemma~\ref{lemma:decay-ratio}}
\label{sec:lemma-decay-ratio-proof}
\begin{proof}
    Let $\left\{\LI_k\right\}_{k=1}^\infty$ be an exhausting sequence of truncations. We pick an arbitrary $\gamma_\ast$ and construct a set 
    \begin{equation}
        \mathbb I(\gamma_\ast) := \{\bm n \in \mathbb N_0^N : \Re \gamma_{\bm n} \leq \gamma_\ast\}.
    \end{equation}
    Since $\Re \gamma_j > 0$ for each $j = 1, \ldots, N$, $\mathbb I(\gamma_\ast)$ is finite. Then, there is $k_0$ such that $\mathbb I(\gamma_\ast) \cap \LJ_k = \varnothing$ for all $k > k_0$, hence $\Gamma'_{\LI_k} > \gamma_\ast$ and $\Gamma'_{\LI_k} \to +\infty$.

    By the definition of the decay rates and since \(R_{\bm n}(\mathbfcal L) \geqslant 0\), one has
    $\Gamma_{\LI} \leqslant \Gamma'_{\LI}$ for any truncation $\LI$. On the other hand, by Lemma~\ref{lemma:gershgorin-radii},
    \begin{equation}
        R_{\bm n}(\mathbfcal L)
        \leqslant
        \sqrt{C(\Re \gamma_{\bm n}+\gamma)},
    \end{equation}
    and the function $x-\sqrt{C(x+\gamma)}$ is increasing for all sufficiently large $x$. Therefore, for all sufficiently large $k$,
    \begin{equation}
        1
        -
        \frac{\sqrt{C(\Gamma'_{\LI_k}+\gamma)}}{\Gamma'_{\LI_k}}
        \leqslant
        \frac{\Gamma_{\LI_k}}{\Gamma'_{\LI_k}}
        \leqslant
        1.
    \end{equation}
    For the sequence of truncations $\{\LI_k\}_{k=1}^\infty$, as $k$ goes to $\infty$, both left-hand and right-hand sides tend to \(1\). By the squeeze theorem,
    ${\Gamma_{\LI_k}}/{\Gamma'_{\LI_k}} \to 1$. In particular, $\Gamma_{\LI_k} \to +\infty$.
\end{proof}

\section{Proofs of theorems and lemmas for truncated Liouvillian}
\subsection{Proof of Lemma~\ref{lemma:schur-convergence}}
\label{sec:lemma-schur-convergence-proof}
\begin{proof}
    First, fix a truncation $\LI$ such that $\Gamma_\LI > 0$ and a point $z$ such that $\Re z>-\Gamma_\LI$. Proposition~\ref{prop:simplest-bound} then applies to $\mathbfcal L_\LJJ$ at both $z$ and $0$.

    By the definitions of \(\mathbfcal S_\LI(z)\) and \(\mathbfcal L_\LI\),
    \begin{equation}
        \mathbfcal S_\LI(z) - z + \mathbfcal L_\LI
        =
        -\mathbfcal L_\LIJ
        \left[
            \left(z - \mathbfcal L_\LJJ\right)^{-1}
            +
            \left(\mathbfcal L'_\LJJ\right)^{-1}
        \right]
        \mathbfcal L_\LJI.
        \label{eq:schur-minus-trunc}
    \end{equation}
    We rewrite the bracket as
    \begin{equation}
        \begin{aligned}
            \left(z - \mathbfcal L_\LJJ\right)^{-1}
            +
            \left(\mathbfcal L'_\LJJ\right)^{-1}
            &=
            \left[
                \left(z - \mathbfcal L_\LJJ\right)^{-1}
                +
                \left(\mathbfcal L_\LJJ\right)^{-1}
            \right]
            +
            \left[
                \left(\mathbfcal L'_\LJJ\right)^{-1}
                -
                \left(\mathbfcal L_\LJJ\right)^{-1}
            \right] \\
            &=
            z\left(z - \mathbfcal L_\LJJ\right)^{-1}\left(\mathbfcal L_\LJJ\right)^{-1}
            -
            \left(\mathbfcal L_\LJJ\right)^{-1}
            \left(
                \mathbfcal L_\LJJ - \mathbfcal L'_\LJJ
            \right)
            \left(\mathbfcal L'_\LJJ\right)^{-1}.
        \end{aligned}
    \end{equation}
    Hence
    \begin{equation}
        \begin{aligned}
            \left \|
            \mathbfcal S_\LI(z) - z + \mathbfcal L_\LI
            \right \|_1
            \leqslant
            \left \| \mathbfcal L_\LIJ \right \|_1
            \left \| \left(\mathbfcal L_\LJJ\right)^{-1} \right \|_1
            \left \| \mathbfcal L_\LJI \right \|_1 &\left[
                |z| 
                \left \|
                    \left(z - \mathbfcal L_\LJJ\right)^{-1}
                \right \|_1 \right . \\ &+ \left. 
                \left \|
                    \left(
                        \mathbfcal L_\LJJ - \mathbfcal L'_\LJJ
                    \right)
                    \left(\mathbfcal L'_\LJJ\right)^{-1}
                \right \|_1
            \right].
        \end{aligned}
        \label{eq:schur-convergence-bound}
    \end{equation}

    We now bound each factor. By Proposition~\ref{prop:simplest-bound},
    \begin{equation}
        \left \|
        \left(z - \mathbfcal L_\LJJ\right)^{-1}
        \right \|_1
        \leqslant
        \frac{1}{\Re z+\Gamma_\LI},
        \qquad
        \left \|
        \left(\mathbfcal L_\LJJ\right)^{-1}
        \right \|_1
        \leqslant
        \frac{1}{\Gamma_\LI},
        \qquad
        \left \|
        \left(\mathbfcal L'_\LJJ\right)^{-1}
        \right \|_1
        \leqslant
        \frac{1}{\Gamma'_\LI}.
    \end{equation}
    Moreover, from the bounds~\eqref{eq:off-diag-block-bounds},
    \begin{equation}
        \left \| \mathbfcal L_\LJI \right \|_1
        \leqslant
        \sqrt{C\left(\Upsilon_\LI + \gamma\right)},
        \qquad
        \left \| \mathbfcal L_\LIJ \right \|_1
        \leqslant
        \sqrt{C\left(\Upsilon_\LI + \tilde \gamma\right)}.
    \end{equation}
    Finally,
    \begin{equation}
        \left \|
            \left(
                \mathbfcal L_\LJJ - \mathbfcal L'_\LJJ
            \right)
            \left(\mathbfcal L'_\LJJ\right)^{-1}
        \right \|_1
        \leqslant
        \sup_{\bm n \in \LJ}
        \frac{R_{\bm n}(\mathbfcal L)}{\Re \gamma_{\bm n}} \leqslant
        \frac{\sqrt{C\left(\Gamma'_\LI + \gamma\right)}}{\Gamma'_\LI} =
        O\left(
            (\Gamma'_\LI)^{-1/2}
        \right),
    \end{equation}
    where the last inequality holds due to Lemma~\ref{lemma:gershgorin-radii}.

    Substituting all these bounds into~\eqref{eq:schur-convergence-bound} gives
    \begin{equation}
        \left \|
        \mathbfcal S_\LI(z) - z + \mathbfcal L_\LI
        \right \|_1
        \leqslant
        \varepsilon_\LI(z),
    \end{equation}
    where $\varepsilon_\LI(z)$ is defined in~\eqref{eq:schur-error-majorant}.

    Now let $\Omega\subset\mathbb C$ be bounded and $\{\LI_k\}_{k=1}^\infty$ be exhausting, and set $z_*:=\sup_{z\in\Omega}|z|$. By Lemma~\ref{lemma:decay-ratio}, $\Gamma_{\LI_k}\to+\infty$, so the explicit bound applies throughout $\Omega$ for all sufficiently large $k$. Using $|z|\leqslant z_*$ and $\Re z+\Gamma_\LI\geqslant\Gamma_\LI-z_*$, we obtain
    \begin{equation}
        \sup_{z\in\Omega}\varepsilon_\LI(z)
        \leqslant
        \frac{
            C\sqrt{(\Upsilon_\LI+\gamma)(\Upsilon_\LI+\tilde\gamma)}
        }{\Gamma_\LI}
        \left[
            \frac{z_*}{\Gamma_\LI-z_*}
            +
            \frac{\sqrt{C(\Gamma_\LI'+\gamma)}}{\Gamma_\LI'}
        \right].
    \end{equation}
    By the definition of an exhausting sequence, $\Upsilon_\LI/\Gamma'_\LI$ is bounded for all sufficiently large truncations, while Lemma~\ref{lemma:decay-ratio} gives $\Gamma_\LI/\Gamma'_\LI\to1$. The right-hand side therefore converges to zero, which proves the claim.
\end{proof}

\subsection{Proof of Lemma~\ref{lemma:unstable-spectrum}}
\label{sec:lemma-unstable-spectrum-proof}
\begin{proof}
    Let $\{\LI_k\}_{k=1}^\infty$ be an exhausting sequence. By its definition there exist $\Lambda$ and $k_0$ such that ${\Upsilon_{\LI_k}} / {\Gamma'_{\LI_k}} \leqslant \Lambda + 1$ for each $k > k_0$.
    Let
    \begin{equation}
        \gamma_{\min}
        :=
        \min_{1\leqslant j\leqslant N}\Re \gamma_j.
    \end{equation}
    By definition, \(\Gamma'_\LI \geqslant \gamma_{\min}\) for every truncation \(\LI\).
    In the rest of the proof, fix $k>k_0$ and write $\LI=\LI_k$ for brevity.

    We first bound the difference between the block \(\mathbfcal L_\LII\) and the truncated Liouvillian \(\mathbfcal L_\LI\):
    \begin{equation}
        \left \|
        \mathbfcal L_\LII - 
        \mathbfcal L_\LI
        \right \|_1 = \left \|
        \mathbfcal L_\LIJ
        \left(\mathbfcal L_\LJJ'\right)^{-1}
        \mathbfcal L_\LJI
        \right\|_1 \leqslant
        \frac{C \left(\Upsilon_\LI + \tilde \gamma\right)}{\Gamma_\LI'}  \leqslant
        \Delta := C\left(\Lambda + 1 + \frac{\tilde \gamma}{\gamma_{\min}}\right).
    \end{equation}
    In particular, for every \(\bm n \in \LI\),
    \begin{equation}
        \left\|
        \mathcal L_{\bm n\bm n}
        -
        \left(\mathcal L_\LI\right)_{\bm n\bm n}
        \right\|
        \leqslant
        \Delta.
        \label{eq:diag-diff-delta}
    \end{equation}

    We now bound the Gershgorin radii of \(\mathbfcal L_\LI\). For \(\bm n \in \LI\),
    \begin{equation}
        \begin{aligned}
            R_{\bm n} \left(\mathbfcal L_\LI\right) &\leqslant
            R_{\bm n} \left(\mathbfcal L_\LII\right) + \sum\limits_{\bm m \in \LI \setminus \{\bm n\}} \left \|
            \sum\limits_{\bm k \in \LJ}
            \mathcal L_{\bm{mk}}
            \mathcal L_{\bm{kk}}^{-1}
            \mathcal L_{\bm{kn}}
            \right \| \\
            &\leqslant
            R_{\bm n} \left(\mathbfcal L\right) +
            \sum\limits_{\bm k \in \LJ}
            \left \| \mathcal L_{\bm{kk}}^{-1} \right \|
            \left \| \mathcal L_{\bm{kn}} \right \|
            \sum\limits_{\bm m \in \LI} 
            \left \| \mathcal L_{\bm{mk}} \right \| \\
            &\leqslant
            R_{\bm n} \left(\mathbfcal L\right) + 
            \sum\limits_{\bm k \in \LJ}  
            \left \| \mathcal L_{\bm{k n}} \right\|
            \frac{\sqrt{C(\Re \gamma_{\bm k} + \gamma)}}{\Re \gamma_{\bm k}} \\
            &\leqslant
            R_{\bm n} \left(\mathbfcal L\right) + 
            \sqrt{C(\Upsilon_\LI + \gamma)}
            \frac{\sqrt{C(\Gamma'_\LI + \gamma)}}{\Gamma'_\LI} \\
            &\leqslant
            R_{\bm n} \left(\mathbfcal L\right) + 
            \frac{C(\Upsilon_\LI + \tilde\gamma)}{\Gamma'_\LI}
            \leqslant R_{\bm n}(\mathbfcal L) + \Delta.
        \end{aligned}
    \end{equation}

    Since \(\mathbfcal L_\LI\) acts on a finite-dimensional space, Theorem~\ref{th:gershgorin-hard} yields
    \begin{equation}
        \sigma(\mathbfcal L_\LI) \subset \bigcup_{\bm n \in \LI}
        \left \{
            z \in \mathbb C: 
            \beta\left(
                z; \left(\mathcal L_\LI\right)_{\bm{nn}}
            \right) \leqslant R_{\bm n}(\mathbfcal L) + \Delta
        \right \}.
    \end{equation}
    For each $z \in \mathbb C$ and $\bm n \in \LI$, we have
    \begin{equation}
        \beta\left(z; \left(\mathcal L_\LI\right)_{\bm{nn}}\right) = 
        \inf_{\| \hat \rho \| = 1} \left\|
            \left[
                z - \mathcal L_{\bm{nn}}
                - \left(\mathcal L_\LI\right)_{\bm{nn}}
                +  \mathcal L_{\bm{nn}}
            \right]\hat \rho
        \right\| \geqslant 
        \beta \left(z; \mathcal L_{\bm{nn}}\right) - \Delta.
    \end{equation}
    Therefore,
    \begin{equation}
        \sigma(\mathbfcal L_\LI) \subset \widetilde{G}(\mathbfcal L) :=
        \bigcup_{\bm n}\widetilde{G}_{\bm n}(\mathbfcal L),
    \end{equation}
    where
    \begin{equation}
        \widetilde{G}_{\bm n}(\mathbfcal L) :=
        \left \{
            z \in \mathbb C: 
            \beta(z; \mathcal L_{\bm{nn}})
            \leqslant R_{\bm n}(\mathbfcal L) + 2 \Delta
        \right \}.
    \end{equation}
    Here we also used Lemma~\ref{lemma:inv-diagonal}. The same inclusion holds for the full Liouvillian,
    $\sigma(\mathbfcal L) \in \widetilde G(\mathbfcal L)$,
    because \(R_{\bm n}(\mathbfcal L) \leqslant R_{\bm n}(\mathbfcal L)+2 \Delta\).

    We now define
    \begin{equation}
        G_+(\mathbfcal L) := \widetilde G(\mathbfcal L) \cap \left \{
            z \in \mathbb C : \Re z \geqslant 0
        \right\}.
    \end{equation}
    Then,
    $\sigma_+(\mathbfcal L) \subset G_+(\mathbfcal L)$
    and
    $\sigma_+(\mathbfcal L_{\LI_k}) \subset G_+(\mathbfcal L)$ for every $k>k_0$.
    Finally, \(G_+(\mathbfcal L)\) is compact. Indeed,
    \begin{equation}
        R_{\bm n}(\mathbfcal L) + 2\Delta
        \leqslant
        \sqrt{C(\Re \gamma_{\bm n}+\gamma)} + 2\Delta
        =
        o(\Re \gamma_{\bm n})
    \end{equation}
    as \(\Re \gamma_{\bm n}\to+\infty\). Hence, only finitely many sets \(\widetilde G_{\bm n}(\mathbfcal L)\) intersect the closed right half-plane \(\Re z \geqslant 0\). Therefore, \(G_+(\mathbfcal L)\) is a finite union of compact sets, and is itself compact.
\end{proof}

\bibliography{main}

\end{document}